\documentclass[12pt]{article}

\usepackage{epstopdf,amsfonts,amsmath}
\usepackage{graphicx}
\usepackage{tikz}

\DeclareGraphicsExtensions{.eps}

\newcommand\encadremath[1]{\vbox{\hrule\hbox{\vrule\kern8pt
\vbox{\kern8pt \hbox{$\displaystyle #1$}\kern8pt}
\kern8pt\vrule}\hrule}}
\def\enca#1{\vbox{\hrule\hbox{
\vrule\kern8pt\vbox{\kern8pt \hbox{$\displaystyle #1$}
\kern8pt} \kern8pt\vrule}\hrule}}

\newcommand\figureframex[3]{
\begin{figure}[bth]
\hrule\hbox{\vrule\kern8pt
\vbox{\kern8pt \vbox{
\begin{center}
{\mbox{\epsfxsize=#1.truecm\epsfbox{#2}}}
\end{center}
\caption{#3}
}\kern8pt}
\kern8pt\vrule}\hrule
\end{figure}
}
\newcommand\figureframey[3]{
\begin{figure}[bth]
\hrule\hbox{\vrule\kern8pt
\vbox{\kern8pt \vbox{
\begin{center}
{\mbox{\epsfysize=#1.truecm\epsfbox{#2}}}
\end{center}
\caption{#3}
}\kern8pt}
\kern8pt\vrule}\hrule
\end{figure}
}

\makeatletter
\@addtoreset{equation}{section}
\makeatother
\newtheorem{theorem}{Theorem}[section]
\newtheorem{conjecture}{Conjecture}[section]
\newtheorem{remark}{Remark}[section]
\newtheorem{proposition}{Proposition}[section]
\newtheorem{lemma}{Lemma}[section]
\newtheorem{corollary}{Corollary}[section]
\newtheorem{definition}{Definition}[section]
\newtheorem{example}{Example}[section]

\newcommand{\diff}{\operatorname{d}}

\newcommand{\Det}{\operatorname{Det}}

\newcommand{\Symb}{\operatorname{Symb}}

\newcommand{\Res}{\mathop{\,\rm Res\,}}
\usepackage[pdftex]{hyperref}
\usepackage{stmaryrd}
\hypersetup{colorlinks,urlcolor=magenta,citecolor=red,linkcolor=blue,filecolor=black}

\begin{document}

\sloppy

\hfill IPhT-T17/180 , CRM-3365

\addtolength{\baselineskip}{0.20\baselineskip}
\begin{center}
\vspace{1cm}

{\Large \bf {Integrability of $\mathcal W({\mathfrak{sl}_d})$-symmetric Toda conformal field theories I : Quantum geometry
}}

\vspace{1cm}

{Rapha\"el Belliard}$^1$,
{Bertrand Eynard}$^{2,3}$,

\vspace{5mm}
$^1$\  Deutsches Elektronen-Synchrotron Theorie Group,
\\
Notkestrasse 85, 22607 Hamburg, Germany
\\
$^2$\ Institut de Physique Th\'eorique, Universit\'e Paris Saclay, 
\\
CEA, CNRS, F-91191 Gif-sur-Yvette, France
\\
$^3$\ Centre de Recherches Math\'ematiques, Universit\'e de Montr\'eal,
\\
2920, Chemin de la tour, 5357 Montr\'eal, Canada.
\vspace{5mm}
\end{center}

\begin{center}
{\bf Abstract}

In this article which is the first of a series of three, we consider $\mathcal W({\mathfrak{sl}_d})$-symmetric conformal field theory in topological regimes for a generic value of the background charge, where $\mathcal W({\mathfrak{sl}_d})$ is the W-algebra associated to the affine Lie algebra $\widehat{\mathfrak{sl}_d}$, whose vertex operator algebra is included to that of the affine Lie algebra $\widehat{\mathfrak g}_1$ at level 1. In such regimes, the theory admits a free field representation. We show that the generalized Ward identities assumed to be satisfied by chiral conformal blocks with current insertions can be solved perturbatively in topological regimes. This resolution uses a generalization of the topological recursion to non-commutative, or quantum, spectral curves. In turn, special geometry arguments yields a conjecture for the perturbative reconstruction of a particular chiral block.

\end{center}
\begin{quote}

\end{quote}


\tableofcontents


\section{Lightning review of conformal field theory}

Conformal field theories in two dimensions have appeared in the physics literature as powerful tools to study numerous systems, from critical (possibly quantum) statistical models in two dimensions (in some thermodynamic limit) to the world-sheet conformal symmetry of string theories \cite{BPZ}. They have the particularity to exhibit infinite dimensional conformal algebras of symmetries, namely extensions of the Virasoro algebra (that are not necessarily Lie algebras as we shall see). It was in turn argued in some cases \cite{BL} that there exists an underlying structure of quantum integrable system with commuting transfer matrices and such.

Mathematically speaking, all these constructions assume the existence of a certain set of functions called $M$-points correlation functions, for some integer $M\in\mathbb N^*$, defined on $M$ copies of a given connected Riemann surface $\Sigma$, and denoted formally as 

\begin{eqnarray}
\left\langle\, \prod_{j=1}^M \Phi_j(p_j)\, \right \rangle
\end{eqnarray}

for distinct generic points $p_1,\dots, p_M\in\Sigma$ called the punctures. They are defined as solutions to linear differential equations called Ward identities and are assumed to be smooth on the generic locus of $\Sigma^M$ and to satisfy a set of axioms, written here for our purpose.

\begin{itemize}

\item \emph{Axiom 1 : Holomorphic factorization.}
For any $M\in\mathbb N^*$, much like a Hodge decomposition, or a separation of variables, there exists a sequence of objects called conformal blocks $\{\mathcal F_\gamma\}_{\gamma\in B_{\mathcal G}}$ such that 

\begin{eqnarray}
\left\langle\, \prod_{j=1}^M \Phi_j(p_j)\,\right\rangle = \sum_{\gamma,\gamma'\in B_{\mathcal G}} C^{\gamma,\gamma'}\mathcal F_\gamma(\bold z)\, \mathcal F_{\gamma'}(\bold{\overline z})
\end{eqnarray}

where we introduced the vector notation $\bold z =(z_1,\dots, z_M)$. $B_{\mathcal G}$ is the set of labels parametrizing this basis of conformal blocks and it contains in particular the data of $\mathcal G$, a channel, namely a certain choice of unicellular trivalent graph on the considered Riemann surface satisfying $\partial\mathcal G=\{z_1,\dots, z_M\}$ and $\pi_1(\Sigma-\mathcal G,o)=0$ with respect to a chosen reference point $o\in\Sigma$. This axiom allows to reduce the problem to its holomorphic (often called chiral) and anti-holomorphic (anti-chiral) parts.

\end{itemize}

In Physics, one wishes the correlation functions that are reconstructed in this way to be modular invariant. It is known that in the case of $\mathfrak g=\mathfrak{sl}_2$, elements of a basis of conformal blocks are  labeled by simply laced (ADE) Dynkin diagrams, in turn these diagrams classify admissible modular invariant correlation functions. Unfortunately, such a statement does not exist yet for higher rank Lie algebras.

In this article we will be interested solely in studying chiral Toda conformal field theory, or holomorphic conformal blocks of the W-algebra $\mathcal W({\mathfrak{sl}_d})$, denoted generically $\mathcal F_\gamma(\bold z)=\left\langle \prod_{j=1}^M V_{\alpha_j}(z_j) \right\rangle$, where we choose the vertex operators $V_{\alpha_j}$'s are to be primary fields of the W-algebra algebra we will soon introduce. Namely, we will assume operator product expansions with the W-generators of the form

\begin{eqnarray}
\bold W^{(k)}(x)V_{\alpha_j}(z_j) & \underset{x\rightarrow z_j}{=} & \frac{q^k(\alpha_j)}{(x-z_j)^k}V_{\alpha_j}(z_j)\ +\ \mathcal O\left(\frac 1{(x-z_j)^{k-1}} \right)
\end{eqnarray}

with Weyl-invariant leading coefficient $q^k(\alpha_j)$. To compute such amplitudes (and as is customary, in quantum mechanics, to mimic the interaction of an observer with the system) we introduce a probe, a so-called chiral spin-one current $\bold J(\widetilde{x})$ (understand locally holomorphic) valued in the dual $\mathfrak{g}^*$ of the Lie algebra and defined for points $\widetilde{x}\in\widetilde{\Sigma}$ in the universal cover $\widetilde{\Sigma}\longrightarrow\Sigma$. It can be seen as being multi-valued on $\Sigma$ and generically having singularities at points of the universal covering that project to any of the $z_j$'s. This relates to the fact that we will assume the algebra of symmetry of be smaller than that generated by the current.

In this quantum theory, the vertex operators are interpreted as the matter content with which the current interacts. This interaction is such that to configurations of points on the Riemann surface, where the operators and currents are inserted, are associated  correlations, describing the entanglement of the particles.

The following axioms are analytic and algebraic requirements these correlations should satisfy as functions of these configurations of points.

Let us fix once and for all the Lie algebra we consider to be $\mathfrak g\underset{def}{=}\mathfrak{sl}_d$ and choose a set of simple roots $\mathfrak R_0\underset{def}{=}\{\mathfrak e_1,\dots,\mathfrak e_{d-1}\}\in\mathfrak h^*$, denoting generically a Cartan sub-algebra by $\mathfrak h\subset\mathfrak g$. We will denote by $\mathfrak R_+$ (resp. $\mathfrak R_-$) the corresponding set of positive (resp. negative) roots. Introducing the \textit{minimal} invariant bilinear form $(\cdot,\cdot)$ on $\mathfrak g^*\times \mathfrak g^*$ (giving length 2 to simple roots), let us consider the algebra generated by a central element $K$ together with the harmonics $(\bold J^{(n)})_{n\in\mathbb Z}$, or modes, obtained by decomposing the chiral current around any generic point $\widetilde x_0 \in \widetilde \Sigma$ (with local coordinate $t=x-x_0$) as

\begin{eqnarray}
\bold J(\widetilde x) \underset{def}{=}\sum_{n\in\mathbb Z}\frac{\bold J^{(n)}(\widetilde x_0)}{(x-x_0)^{n+1}}
\end{eqnarray}

\begin{remark}
We will drop the explicit writing of the dependence of the modes in the generic point $x_0$ when no confusion is possible. Modes can only be compared when evaluated at the same point.
\end{remark}

These generators satisfy the commutation relations

\begin{eqnarray}
[\bold J^{(n)},\bold J^{(m)}] =[\bold J,\bold J]^{(n+m)}+n(\bold J^{(n)},\bold J^{(m)})\delta_{n+m,0} K
\end{eqnarray}

where the $\bold J$ symbols in $[\bold J,\bold J]$ are generically evaluated at different Lie algebra elements and therefore have non-trivial Lie bracket. The so-called affine algebra at level $\kappa\in\mathbb C$ denoted $\widehat{\mathfrak g}_\kappa$ is then defined as the Lie algebra $\widehat{\mathfrak g}_\kappa \underset{def}{=}\widehat{\mathfrak g}\, \big/(\kappa-K)$, where $\widehat{\mathfrak g}$, called the generic affine algebra associated to $\mathfrak g$, is defined as the central extension of vector spaces

\begin{eqnarray}
0\longrightarrow \mathbb C K\longrightarrow \widehat{\mathfrak g}_\kappa \longrightarrow \mathcal L(\mathfrak g)\oplus\mathbb C \partial \longrightarrow 0
\end{eqnarray}

where $\mathcal L(\mathfrak g)$ is the \textit{loop algebra} of $\mathfrak g$ denoted $\mathcal L(\mathfrak g)\underset{def}{=}\mathfrak g((t))$ (endowed with the natural Lie algebra structure coming from $\mathfrak g$) and the extra generator $\partial$ is defined to satisfy

\begin{eqnarray}
[\partial, M]=\frac{\text{d}}{\text{dt}}M,\quad (\text{and thus}\,\, [\partial, K]=0)
\end{eqnarray}

for any $M\in\mathcal L(\mathfrak g)=\mathfrak g((t))$. $K$ is a central element assumed to act trivially as multiplication by $\kappa$ (using the fact that the short exact sequence splits, this is equivalent to focusing on diagonalizable $\widehat{\mathfrak g}$-module with finite weight spaces \cite{Her} and restricting ourselves to the highest weight representations they define). The levels $\kappa$ are in general not constrained but they are for example in the case where the conformal field theory can be extended to a certain three dimension topological field theory named Chern-Simons theory on a given 3-manifold $M$ whose boundary $\partial M=\Sigma$ is a Riemann surface. The levels are then often required to make two copies of Chern-Simons on $M$ equivalent if they yield the same conformal field theory on $\partial M$. They are then parametrized by maps $M_\Sigma M\longrightarrow G$ considered up to homotopy, where $M_\Sigma M$ denotes gluing of the copies of $M$ along their identical boundary $\Sigma$ with matching of orientations. In the case where $G=SU(2)$, $\Sigma$ is the Riemann sphere and $M\subset\mathbb R^3$ is the unit ball, since gluing in this case yields a $3$-sphere, the levels of the affine algebras of interest are then parametrized by the third homotopy group given by $\pi_3(SU(2))\simeq \mathbb Z$.  Constraints can also arise from the representation theory of the affine algebra, indeed, when the levels under consideration are positive integers, $\widehat{\mathfrak g}$ admits unitary highest weight representations whose highest weights are dominant integral (quantization condition).

The vertex operator algebra corresponding to the W-algebra $\mathcal W({\mathfrak{sl}_d})$ will be included in the one of the affine Lie algebra at level 1 denoted $\widehat{\mathfrak g}_1$. This latter vertex operator algebra then appears as a module over the former, allowing to consider the following

\begin{definition}{Insertions of currents}\\
By insertions of currents into chiral correlation functions we mean that we consider an infinite yet countable set of additional $(\mathfrak g^*)^{\otimes n}$-valued functions of interest denoted

\begin{eqnarray}
 \big\langle \big\langle \bold J(\widetilde x_1)\cdots \bold J(\widetilde x_n)\big\rangle\big\rangle\underset{def}{=}\frac{\left\langle\, \bold J(\widetilde x_1)\cdots \bold J(\widetilde x_n)\, \prod_{j=1}^M V_{\alpha_j}(z_j)\, \right\rangle}{\left\langle\,  \prod_{j=1}^M V_{\alpha_j}(z_j)\, \right\rangle}
\end{eqnarray}

\end{definition}

A set of function that are required to be compatible with the following

\pagebreak 

\begin{itemize}

\item \emph{Axiom 2 : Operator product expansion.}

Keeping the notation $(\cdot,\cdot)$ for the form on $\mathfrak g\times\mathfrak g$ dual to minimal invariant bilinear form on $\mathfrak g^*\times\mathfrak g^*$,

\begin{eqnarray}
\bold J(\widetilde x\cdot E)\bold J(\widetilde y\cdot F) &\underset{x\sim y}{=} &-\frac{(E,F)}{(x-y)^2}\, +\, \frac{\bold J(\widetilde y\cdot [E,F])}{x-y} +\ \mathcal O(1)
\end{eqnarray}

Similarly, a way to realize the type of vertex operators under consideration is to assume, for any puncture index $j\in\{1,\dots,M\}$, the existence of a Cartan sub-algebra $\mathfrak h_j \subset \mathfrak g$ and an element $\alpha_j\in\mathfrak h_j^*$ such that

\begin{eqnarray}
\bold J(\widetilde x\cdot E)\, V_{\alpha_j}(z_j)\underset{x\sim z_j}{=} - \frac{\alpha_j(E)}{x-z_j}V_{\alpha_j}(z_j)\ +\ \mathcal O(1)
\end{eqnarray}

for $\widetilde x,\widetilde y \in\widetilde\Sigma$ and some Cartan elements $E,F\in\mathfrak h_j$. We also introduced the linear  notation $\bold J(\widetilde x \cdot E)\underset{def}{=}\bold J(\widetilde x)(E)$ to relate with our notations in the study of Fuchsian differential systems for the evaluation \cite{BEM1}, \cite{BER1}. We denote the normal ordering operation hidden in the $\mathcal O(1)$ symbols as $:A(\widetilde x)B(\widetilde y) \underset{x=y}{:}$ defined as the next to singular term when the base-points of $\widetilde x,\widetilde y\in\widetilde\Sigma$ come together (but are not necessarily such that $\widetilde x=\widetilde y$).

The elements $\alpha_j\in\mathfrak h_j^*$  appearing in the last asymptotic equality are constrained such that vertex operators $V_{\alpha_j}$'s have the right Weyl invariant leading coefficients $q^k(\alpha_j)$ with W-algebra generators $\bold W^{(k)}$'s when using the definitions of next section. Because we will define the W-algebra by its free field realization $\mathcal W({\mathfrak{sl}}_d)\subset\overline{\mathcal U}(\widehat{\mathfrak h})$, this amounts to view a module over a vertex operator algebra as a module over one of its sub-vertex operator algebras.

The presence of simple poles moreover mean that we consider only regular singularities. Irregular singularities in the $\mathfrak g=\mathfrak{sl}_2(\mathbb C)$ case of Liouville conformal field theory were studied in \cite{GT}.
\end{itemize}

\begin{definition}{Background charge}\\
Define the background charge as the number denoted $Q\underset{def}{=}b+b^{-1}$ and parametrized by the non-zero complex number $b\in\mathbb C^*$.
\end{definition}

These asymptotic relations are to be understood as holding when inserted into correlation functions, that is to say that they are meromorphic conditions on the functions we denoted $\langle\langle\bold J\cdot\dots\cdot\bold J\rangle\rangle$. They are strong requirements as $\bold J$ contains for example both data of the Lie bracket and the minimal invariant bilinear form on $\mathfrak g$.

Recall that the Virasoro algebra $Vir$ is the infinite dimensional Lie algebra that generates the conformal transformations of the complex plane. It is defined as the central extension of vector spaces

\begin{eqnarray}
0\longrightarrow \mathbb C c \longrightarrow Vir \longrightarrow \text{Der}_{\mathbb C} \longrightarrow 0
\end{eqnarray}

where we introduced the Lie algebra $\text{Der}_{\mathbb C}$ of holomorphic derivations of the field of Laurent series on the complex plane as well as the central element acting  the scalar $c\underset{def}{=}d-1+12Q^2$ ($c$ stands for Casimir) called the central charge. $Vir$ is generated by $c$ together with the elements $(L_n)_{n\in\mathbb Z}$ satisfying the famous commutation relations

\begin{eqnarray}
[L_n,L_m]=(n-m)L_{n+m}+\frac {c}{12}n(n^2-1)\delta_{n+m,0},\quad (n,m)\in\mathbb Z^2
\end{eqnarray}

where for any $n\in\mathbb Z$, $L_n$ generates the one-parameter family of local  conformal transformations $(z\longmapsto t\, z^{n+1})_{t\in\mathbb C}$ e.g. $L_0$ is the dilation operator. It goes to the Witt algebra in the zero central charge limit $c\longrightarrow 0$. 
By essence of conformal field theory, generators of the Virasoro algebra can be gathered into a meromorphic stress-energy tensor that belongs to the vertex operator algebra under consideration (anticipating on the next axiom by denoting the variable by $x$ and not one of its pre-images $\widetilde x$ by the universal covering map). It is such that it can be expanded around a basepoint $x_0\in\Sigma$ as

\begin{eqnarray}
T(x)\underset{x\sim x_0}{=}\sum_{n\in\mathbb Z}\frac{L_n(x_0)}{(x-x_0)^{n+2}},
\end{eqnarray}

and the Virasoro canonical commutation relations are translated in the following \emph{operator product expansion}

\begin{eqnarray}
T(x')T(x)\underset{x'\sim x}{=}\frac{c/2}{(x'-x)^4}\, +\, \frac{2\, T(x)}{(x'-x)^2}\, +\, \frac{\partial T(x)}{x'-x}\, +\, \mathcal O(1)
\end{eqnarray}

Similarly, the operator product expansion of the stress-energy tensor with the chiral current is defined to be

\begin{eqnarray}
T(x)\bold J(\widetilde y) \underset{x\sim y}{=}\frac{\bold Q}{(x-y)^3}+\frac{\bold J(\widetilde y)}{(x-y)^2}+\frac{\partial\bold J(\widetilde y)}{x-y}+\mathcal O(1)
\end{eqnarray}

where $\bold Q\underset{def}{=} Q\rho$, $\rho\underset{def}{=}\frac 12\sum_{\mathfrak r \in \mathfrak R_+} \mathfrak r$ being the Weyl vector, to be again understood as identities holding when inserted into correlation functions. The coefficient $1$ in front of the second order pole in the last expression tells us that the current $\bold J$ has spin (conformal weight) $1$.

\begin{remark}
One might be afraid that such a decomposition for the stress-energy tensor would create singularities of infinite order in some operator product expansion appearing in the theory but a requirement of the vertex operator algebra formalism is that any admissible field $V_\alpha$ should be annihilated by all high enough modes of $T$, see \cite{Bor} for details. In particular, define an admissible ground state as a vector $|0\rangle\in\mathcal A$ in the considered representation satisfying the so-called \textit{Virasoro constraints}

\begin{eqnarray}
\forall n\geq -1,\qquad L_n|0\rangle=0
\end{eqnarray}

In particular, if we were to assume that $L_n^\dagger=L_{-n}$, then the Virasoro constraints would yield that the expected value of the stress-energy tensor vanishes

\begin{eqnarray}
\langle 0|T(x)|0\rangle=0
\end{eqnarray}

namely that we have conformal symmetry in this ground state at the quantum level. We will not however be assuming the existence of such a ground state in our study.
\end{remark}

The next axiom is at the heart of the method we adopt to study conformal field theories. As was mentioned in the introduction, a path integral formulation of the problem with a Lagrangian allows for the derivation of Schwinger-Dyson equations. Their counterparts in this non-perturbative definition of conformal field theories are the following conformal Ward identities. 

\begin{itemize}

\item \emph{Axiom 3 : Generalized conformal Ward identities.}

For any generic $\widetilde{x}_1,\cdots,\widetilde{x}_n\in\widetilde\Sigma$, $\big\langle\big\langle T(x)\bold J(\widetilde x_1)\cdots\bold J(\widetilde x_n)\big\rangle\big\rangle$ is a holomorphic function of the variable $x\in\Sigma-\{z_1,\dots,z_M,x_1,\dots,x_n\}$ with meromorphic singularities at $x=x_i$ and $x=z_j$ whose behaviors are prescribed by the operator product expansions.

\end{itemize}

This is again an axiom prescribing some analytic conditions for the functions of interest. We will be applying similar ideas for the generating series of generators of $\mathcal W({\mathfrak{sl}_d})$ for which the conformal Ward identities together with the operator products expansions yield the so-called loop equations.

The two last axioms deal with how one can reconstruct the full theory from its chiral and anti-chiral parts. We will not be needing them in the context of this work but we still state them for completeness.

\begin{itemize}

\item \emph{Axiom 4 : Single-valuedness.}

The $M$-points correlation functions have no monodromy around cycles in the moduli space of configurations of $M$ distinct points on the Riemann surface $\Sigma$.

\item \emph{Axiom 5 : Fusion and crossing symmetries.}

The decomposition of the real correlation functions in terms of the conformal blocks requires in particular a choice of channel, a unicellular trivalent graph, on the base Riemann surface $\Sigma$ and different choices of such channels should lead to the same correlation function after reconstruction. This is often referred to as the \textit{associativity} of the operator product expansions.

\end{itemize}

There is no general proof that all these axioms are actually compatible. We will therefore proceed by necessary condition, assuming these axioms to be compatible and satisfied, to define the algebra $\mathcal W({\mathfrak{sl}_d})$ and the so-called insertions of W-generators in these chiral correlation functions with currents. In turn, this will yield $ \mathcal W({\mathfrak{sl}_d})$-symmetric, or Toda, generalized conformal Ward identities. In the second part we will define the associated quantum geometry through the quantum spectral curve. This will turn out to be the deformed initial data needed to run the topological recursion of \cite{EO} in this context and we will show that it constructs perturbatively solutions to the $\mathcal W({\mathfrak{sl}_d})$-symmetric conformal Ward identities. This work is a direct generalization of \cite{CER} where the three points function of Liouville theory on the sphere was checked to be computed by this method to first orders. Let us stress furthermore that we are consider conformal blocks at generic values of their parameters and that in turn, the quantum branch points to be defined will be simple. This implies in particular that the formula for the topological recursion will not need to encompass higher ramification profiles as is for instance taken into account in \cite{BouE}. In the sequel to this paper we shall exhibit an explicit realization of this framework using $\beta$-deformed two-matrix models and extend the formalism to conformal field theories on higher genus surfaces.

\section{W-algebras and associated conformal field theories}

\subsection{From Virasoro to W-algebras}

A conformal field theory \cite{T} is a quantum field theory defined on a Riemann surface $\Sigma$ and endowed with an action of the product $\mathcal A \times \mathcal A'$ of two extensions $Vir\subset\mathcal \mathcal A$, $Vir\subset \mathcal A'$ of the Virasoro algebra (they need not be the same).

Let us stress at this point that these two extensions $\mathcal A$ and $\mathcal A'$ act respectively upon the holomorphic and the anti-holomorphic dependence of the observables defined on the Riemann surface. We will here only be interested in the chiral theory, that is in the action of $\mathcal A$ and in the meromorphic properties of the chiral correlation functions.

We will be interested particularly in the extension $Vir\subset\mathcal A\underset{def}{=}\mathcal W({\mathfrak{sl}_d})$ defined from the affine Lie algebra $\widehat{\mathfrak{sl}_d}$ (fix once and for all $\mathfrak g\underset{def}{=}\mathfrak{sl}_d$), using a higher rank generalization of the Sugawara construction \cite{Su68}, namely the quantum Miura transform, and defining generating functions whose modes generate  $\mathcal W({\mathfrak{sl}_d})$. Let us mention that our method is not directly generalizable to general reductive complex Lie algebra as it relies on the explicit expression of the W-generators that the quantum Miura transform yields and that such a definition does not work in the general case where one has to quantize the Poisson algebra underlying the Drinfeld-Sokolov hierarchy associated to the Lie algebra under consideration (equivalent to the quantum Drinfeld-Sokolov reduction) \cite{BS},\cite{FF}.

The idea behind W-algebras is that they allow for a better encoding of some representations of $Vir$. Indeed, there are spaces representing both the W-algebra and the Virasoro algebra that decompose as an infinite direct sums of irreducible representations of $Vir$ but as a finite direct sums of irreducible representations of $ \mathcal W({\mathfrak{sl}_d})$. In particular, the operator product expansions they satisfy should be expressible in terms of these generators only. We will see two different situations in which this is possible but we will not get any further in studying the representation theory of W-algebras and refer the reader the \cite{BS}, \cite{BW}.

The definition of the W-algebra generators involves non-commutative geometry and $\mathcal W({\mathfrak{sl}_d})$ appears as a subalgebra $\mathcal W({\mathfrak{sl}_d})\subset \overline{\mathcal U}\left(\widehat{\mathfrak g}_1\right)$ of a completion of the universal enveloping algebra of the affine  algebra $\widehat{\mathfrak{sl}}_d$ at level 1.

The background charge plays the role of a quantization parameter and noticing that the W-algebra for generic values of $Q$ reduces to a Casimir algebra in the limit $Q\longrightarrow 0$ will allow for the interpretation of $\mathcal W({\mathfrak{sl}_d})$-symmetric conformal field theory as a quantization of a corresponding $\mathfrak{sl}_d$ Fuchsian system \cite{BEM1},\cite{BER1}.

Before giving the precise definitions, let us review a few generalities on W-algebras.

\subsection{Operator product expansions}

We will throughout this text consider the Lie algebra $\mathfrak{sl}_d$ in its fundamental representation $\mathfrak{sl}_d\subset\mathfrak{gl}_d$. Similarly to the case of the Virasoro algebra, introducing the rank $d-1=\mathfrak{r}k\,\mathfrak{sl}_d$,  the soon to be defined generators of $\mathcal W({\mathfrak{sl}_d})$, denoted $\{\bold W_n^{(d_p)}\}^{n\in\mathbb Z}_{1\leq p\leq d-1}$, fit for a given $p\in\{ 1,\dots, d-1\}$, into a generating function defined around a basepoint $x_0\in\Sigma$ by

\begin{eqnarray}
\bold W^{(d_p)}(\widetilde x)=\sum_{n\in\mathbb Z}\frac{\bold W_n^{(d_p)}}{(x-x_0)^{n+d_p}}\quad \text{for}\, p\in\{ 1,\dots,d-1\}
\end{eqnarray}

where the $d_p$'s are integer indices defined as follows :  since $\mathfrak h$ is a commutative Lie algebra, we have an isomorphism $\mathcal U(\mathfrak h^*)\simeq\mathbb C[\mathfrak h]$ and moreover, by a theorem of Chevalley, the subspace of this last ring which is invariant under the action of the Weyl group is actually a polynomial ring 
$\mathbb C[\mathfrak h]^{\mathfrak w}\simeq\mathbb C[\sigma_1,\dots,\sigma_{d-1}]$ where for any index $p\in\{1,\dots,d-1\}$, we then define $d_p\in\mathbb N^*$ as the degree of the invariant polynomial $\sigma_p$. Since we consider $\mathfrak g=\mathfrak{sl}_d$, we have $\sigma_p=p+1$ for any $p\in\{1,\dots,d-1\}$.

We will assume the algebra $\mathcal W({\mathfrak{sl}_d})$ to be an extension of the Virasoro algebra and will define its generators, denoted $\bold W^{(k)}$, for any $k\in\{2,\dots,d\}$.

Following the introduction of \cite{BW}, the corresponding operator product expansions can be presented schematically as

\begin{eqnarray}
\bold W^{(k)}(\widetilde x)\bold W^{(l)}(\widetilde y) &\underset{x\sim y}{=}& \frac{g^{k,l}}{(x-y)^{k+l}} \nonumber\\
&+& \sum_{s=2}^{d-1} f^{k,l}_{(1),s}\frac{\bold W^{(s)}(\widetilde y)+g^{k,l}_s\partial\bold W^{(s)}(\widetilde y)+\dots}{(x-y)^{k+l-s}} \nonumber\\
&+& \sum_{s,t=2}^{d-1} f^{k,l}_{(2),s,t} \frac{:\bold W^{(s)}(\widetilde y)\bold W^{(t)}(\widetilde y):+\dots}{(x-y)^{p+q-s-t}}+\dots \nonumber\\
\end{eqnarray}

We will identify $\bold W^{(2)}\propto T$ as being the stress energy tensor generating the Virasoro algebra and although we will not need it for this study, the rest of these generating functions, $\bold W^{(k)}$ for $k\neq 2$, could be transformed to primary fields $\widetilde{\bold W}^{(k)}$ of the Virasoro algebra, without changing the W-algebra their mode generate, satisfying the operator product expansions

\begin{eqnarray}
T(x)\widetilde{\bold W}^{(k)}(\widetilde y)\underset{x\sim y}{=}k \frac{\widetilde{\bold W}^{(k)}(\widetilde y)}{(x-y)^2}+\frac{\partial \widetilde{\bold W}^{(k)}(\widetilde y)}{x-y}+\mathcal O(1)
\end{eqnarray}

which imply in particular the commutation relations

\begin{eqnarray}
[L_n,\widetilde{\bold W}^{(k)}_m]=[(k-1)n-m]\widetilde{\bold W}^{(k)}_{n+m}
\end{eqnarray}

Let us mention that for $d=3$, the $\mathfrak g=\mathfrak{sl}_3$ case was investigated in \cite{Zam} and the algebra defined by the corresponding operator product expansions is

\begin{eqnarray}
[L_n,\widetilde{\bold W}_m^{(3)}] & = & (2n-m)\widetilde{\bold W}^{(3)}_{n+m}
\end{eqnarray}

where we identified the modes of $\bold W^{(2)}$ with some Virasoro generators and 

\begin{eqnarray}
[\widetilde{\bold W}_m^{(3)},\widetilde{\bold W}_m^{(3)}]& = & (n-m)[\frac{1}{15} (n+m+2)(n+m+3) - \frac{1}{6} (n+2)(m+2)]L_{n+m} \nonumber \\
&& + \frac{c}{3\cdot 5 !} n(n^2-1)(n^2-4) \delta_{n+m,0} + \frac{16}{22+c} (n-m) \Lambda_{n+m}
\end{eqnarray}

where we introduced the symbols

\begin{eqnarray}
\Lambda_n&\underset{def}{=}&\sum_{k\in\mathbb Z} :L_kL_{n-k}:+\frac 15 \nu_nL_n\\
\text{with}\qquad \nu_{2l}&\underset{def}{=}&(1+l)(1-l)\\  
\ \text{and}\qquad  \nu_{2l+1}&\underset{def}{=}& (2+l)(1-l)
\end{eqnarray}

\subsection{W-algebra generators}

Let us consider the generic situation $Q=b+b^{-1}\neq 0$.

\begin{definition}{W-algebra generators}\\
Consider the weights $h_i=\omega_1 - e_1  \dots - e_{i-1}$, $i=1,\dots,d$ of the first fundamental representation of $\mathfrak{sl}_d$. The generating functions of generators of the algebra $\mathcal W({\mathfrak{sl}_d})$ are expressed in any local coordinates through the quantum Miura transform

\begin{eqnarray}
\bold{\widehat{\mathcal E}}=\sum_{k=0}^d (-1)^{k}  \bold W^{(k)} \widehat y^{\, d-k} \underset{def}{=} \ :\left(\widehat y -\bold J_1\right)\cdots\left(\widehat y-\bold J_d\right):
\end{eqnarray}

where $\widehat y\underset{def}{=}Q\partial$ and for any subscript $i=1,\dots,d$ we defined $\bold J_i\underset{def}{=}(h_i,\widetilde{\bold J})$ (\ $\widetilde{\bold J}\in\mathfrak h^*$ is in the conjugacy class of $\bold J$ and is defined up to conjugation by a Weyl group element). The non-commutative prescription for evaluating these products at coinciding points has been used. These local definitions in coordinate patches have to be glued together to define the fields on $\Sigma$ and could in principle result in multi-valued objects.
\end{definition}

This definition is to be understood as the identification of the coefficients of the vertex operator valued polynomial expression in $\widehat y$ obtained by commuting all the derivative symbols to the right.

\begin{example}
\begin{eqnarray}
\bold W^{(1)}(\widetilde x)=\sum_{i=1}^d \bold J_i(\widetilde x)=0
\end{eqnarray}

\begin{eqnarray}
\bold W^{(2)}(\widetilde x)&=&\sum_{1\leq i<j\leq d} :\bold J_i\bold J_j(\widetilde x):-\, Q\sum_{i=2}^d (i-1)\partial\bold J_i(\widetilde x) \\
&=& -\frac 12 :(\bold J,\bold J)(\widetilde x): +\, (\bold Q,\partial\bold J(\widetilde x))
\end{eqnarray}
where we used $\sum_{i=1}^d h_i =0$ and $\bold Q\underset{def}{=}Q\rho$ together with the expression 
\begin{eqnarray}
\rho=\frac 12 \sum_{i=1}^d(n-2i+1)h_i=-\sum_{i=2}^d (i-1) h_i
\end{eqnarray}
of the Weyl vector.
\end{example}

\medskip

\begin{lemma}

For any $k\in\{ 1,\dots d\}$, the $k^{th}$ generator $\bold W^{(k)}$ of $\mathcal W({\mathfrak{sl}_d})$ is equal to
\begin{eqnarray}
\sum_{p=1}^k (-1)^{k-p}Q^{k-p}\sum_{\underset{k\leq i_p}{1\leq i_1<\dots< i_p\leq d}}\sum_{\underset{p+\sum_{l=1}^p q_l=k}{\overset{\forall l\in\llbracket 1,p\rrbracket}{0\leq q_l\leq i_l-i_{l-1}-1}}} \prod_{l=1}^p\binom{i_l-i_{l-1}-1}{q_l}:\partial^{q_1}\left(\bold J_{i_1}\cdots\partial^{q_p}\bold J_{i_p}\right) : \nonumber\\
&
\end{eqnarray}
\end{lemma}

\begin{proof}{

The proof consists in using non-commutative algebra in the ring of differential operators $\mathcal D_\Sigma$ overs the base curve to commute all $Q\partial$ symbols to the right before identifying the coefficients of the differential operators. To do so we first identify holomorphic functions $f\in\mathcal O_\Sigma$ with the degree $0$ differential operators $f\cdot\in\mathcal D_\Sigma$ of multiplication by $f$ on the left. For any function $f\in\mathcal O_\Sigma$ we then have the commutation relation $[Q\partial,f]=Q(\partial f)$ and it recursively yields the non-commutative version of Leibniz formula

\begin{eqnarray}
(Q\partial)^p f = Q^p\sum_{q=0}^p\binom{p}{q} (\partial^q f)\partial^{p-q}
\end{eqnarray}

where the equality takes place in $\mathcal D_\Sigma$. It is then a straightforward computation to derive the wanted result.
}
\end{proof}

\begin{example}
For $d=2,3$, we get
\begin{eqnarray}
\bullet\quad\bold{\widehat{\mathcal E}}_{d=2}&=&(Q\partial)^2-[\bold J_1+\bold J_2](Q\partial) +:\bold J_1 \bold J_2:-\, Q\left(\partial \bold J_2\right)\\
&=& (Q\partial)^2 - :\bold J_1^2: +\ Q(\partial \bold J_1)\\
&&\nonumber\\
\bullet\quad\bold{\widehat{\mathcal E}}_{d=3}&=&(Q\partial)^3-[\bold J_1+\bold J_2+\bold J_3](Q\partial)^2 \nonumber \\
&+&[:\bold J_1\bold J_2+\bold J_2\bold J_3+\bold J_1\bold J_3:-\, Q\partial \bold J_2-2Q\partial \bold J_3](Q\partial) \nonumber \\ &+&Q^2(\partial^2\bold J_3)+Q:[\bold J_1+\bold J_2]\left(\partial\bold J_3\right):+\ Q:\left(\partial\bold J_2\right)\bold J_3:
\end{eqnarray}
and for any $d\in\mathbb N^*$,

\begin{eqnarray}
\qquad \bold W^{(3)}&=&\sum_{1\leq i<j<k\leq d}:\bold J_i\bold J_j\bold J_k: \nonumber \\
&-& Q\sum_{\underset{3\leq j}{1\leq i<j\leq d}}[(j-i-1):\bold J_i\partial\bold J_j:+\, (i-1)\partial\left(:\bold J_i\bold J_j:\right)] \nonumber \\
&+& Q^2\sum_{i=3}^d\binom{i-1}{2}\partial^2\bold J_i
\end{eqnarray}
\end{example}

For $k\in\{ 2,\dots,d\}$, $\bold W^{(k)}$ therefore involves at most terms of degree $k$ as differential polynomials in $d$ copies of a chosen so-called ``chiral $\mathfrak g^*$-valued spin-one field'' $\bold J(\widetilde x)$ as described before. We require as stated in $Axiom\, 2$ that it satisfies

\begin{eqnarray}
\bold J(\widetilde x\cdot E)\bold J(\widetilde y\cdot F) \underset{x\sim y}{=} -\, \frac{(E,F)}{(x-y)^2}\, +\, \frac{\bold J(\widetilde y \cdot [E,F])}{x-y} +\ \mathcal O(1)
\end{eqnarray}

for Lie algebra elements $E,F\in \mathfrak{sl}_d$, where $(\cdot,\cdot)$ still denotes the corresponding minimal invariant bilinear form.

\subsection{Ward identities}

We now generalize $Axiom\, 3$ to the algebra $\mathcal W({\mathfrak{sl}_d})$ (not necessarily a Lie algebra) defined by the operator coefficients of the expansions of the generators $\bold W^{(k)}$, $k\in\{ 1,\dots d\}$, around a base point $x_0\in\Sigma$.
We then get that the chiral spin-one current $\bold J$ should be chosen such that it satisfies

\begin{definition}
The generalized Ward identities of this $\mathcal W({\mathfrak{sl}_d})$-symmetric conformal field theory are the axiom defined, for any $k\in\{ 2,\dots, d\}$ and any generic points $\widetilde{x}_1,\cdots,\widetilde{x}_n\in\widetilde\Sigma$ in the universal covering, by requiring that  $\big\langle\big\langle \bold W^{(k)}(\widetilde x)\bold J(\widetilde x_1)\cdots\bold J(\widetilde x_n)\big\rangle\big\rangle$ is a holomorphic function of the variable $x\in\Sigma-\{z_1,\dots,z_M,x_1,\dots,x_n\}$ with meromorphic singularities at $x=x_i$ and $x=z_j$ whose behaviors are prescribed by the operator product expansions.
\end{definition}

This definition yields that for an admissible chiral current $\bold J$, the insertion $\langle\langle \bold W^{(k)}(\widetilde x)\bold J(\widetilde x_1)\cdots\bold J(\widetilde x_n)\rangle\rangle$ of any of the  fields $\bold W^{(k)}$, $k\in\{2,\dots,d\}$, should be uniquely valued as a holomorphic function of $x\in\Sigma-\{z_1,\dots,z_M,x_1,\dots,x_n\}$. We can therefore drop the upper-script in $\widetilde x$ and simply write $\bold W^{(k)}(x)$ when evaluating the insertion of such a generator. 

Replacing the previously computed expression for $\bold W^{(k)}$ in terms of the current $\bold J$ in $\langle\langle \bold W^{(k)}(x)\bold J(\widetilde x_1)\cdots\bold J(\widetilde x_n)\rangle\rangle$ yields that

\begin{proposition}{Ward identities as loop equations}
\begin{eqnarray}
\big\langle\big\langle \bold W^{(k)}(x)\bold J(\widetilde x_1)\cdots\bold J(\widetilde x_n)\big\rangle\big\rangle \nonumber \\
=\sum_{p=1}^k (-1)^{k-p}Q^{k-p}&&\sum_{\underset{k\leq i_p}{1\leq i_1<\dots< i_p\leq d}}\sum_{\underset{p+\sum_{l=1}^p q_l=k}{\overset{\forall l\in\llbracket 1,p\rrbracket}{0\leq q_l\leq i_l-i_{l-1}-1}}}\left( \prod_{l=1}^p\binom{i_l-i_{l-1}-1}{q_l}\right) \nonumber \\
&&\qquad\quad\times\quad \big\langle\big\langle:\partial^{q_1}\left(\bold J_{i_1}\dots\partial^{q_p}\bold J_{i_p}\right) (\widetilde x):\bold J(\widetilde x_1)\cdots\bold J(\widetilde x_n)\big\rangle\big\rangle \nonumber \\
&&
\end{eqnarray}
is a holomorphic function of $x\in\Sigma-\{z_1,\dots,z_M,x_1,\dots,x_n\}$.
\end{proposition}

For a generic value of $Q$ and specializing to the cases $k=1,2$, the expressions for $\bold W^{(1)}$ and $\bold W^{(2)}$ yield that

\begin{corollary}
\begin{eqnarray}
\sum_{i=1}^d \big\langle\big\langle \bold J_i(\widetilde x)\bold J(\widetilde x_1)\cdots\bold J(\widetilde x_n)\big\rangle\big\rangle=0 \qquad \text{and}\qquad\qquad\qquad\qquad\\
\sum_{1\leq i<j\leq d}\big\langle\big\langle :\bold J_i\bold J_j(\widetilde x):\bold J(\widetilde x_1)\cdots\bold J(\widetilde x_n)\big\rangle\big\rangle-\, Q\sum_{i=2}^d (i-1)\partial_x\big\langle\big\langle\bold J_i(\widetilde x)\bold J(\widetilde x_1)\cdots\bold J(\widetilde x_n)\big\rangle\big\rangle\nonumber\\
\end{eqnarray}

is a holomorphic function of $x\in\Sigma-\{z_1,\dots,z_M,x_1,\dots,x_n\}$.
\end{corollary}

\subsection{Classical limit $Q\longrightarrow 0$ and quantization}

The definition of the generators of $\mathcal W({\mathfrak{sl}_d})$ involved identifying the coefficients of two differential operators. 

\begin{definition}
The Casimir algebra $\mathcal W_0({\mathfrak{sl}_d})$  is the associative algebra generated by the modes $\{\bold W_{0;n}^{(k)}\}_{\overset{n\in\mathbb Z}{1\leq k\leq d}}$ of the generators defined in any local coordinates through the identities : 

\begin{eqnarray}
\sum_{k=0}^d (-1)^k\bold W_0^{(k)} y^{r-k} &=& :\left(y-\bold J_1\right)\cdots\left(y-\bold J_d\right):\\
\forall p\in\{1,\dots,d\},\quad \bold W_0^{(k)}(x) &\underset{def}{=}&\sum_{n\in\mathbb Z}\frac{\bold W_{0;n}^{(k)}(x_0)}{(x-x_0)^{n+k}}
\end{eqnarray}
\end{definition}

We get the following classical limit

\begin{theorem}{Quantization of Fuchsian differential systems}\\

\begin{eqnarray}
\sum_{k=0}^d (-1)^k\bold W_0^{(k)} y^{r-k} = \Symb\left(\sum_{k=0}^d (-1)^k \bold W^{(k)} \widehat y^{\, r-k}\right)
\end{eqnarray}

where the symbol of a differential operator $P(x,\widehat y\, )\in\mathbb C(x)[\, \widehat y\, ]$ is defined as

\begin{eqnarray}
\Symb\left(P(x,\widehat y\, )\right) \underset{def}{=} \underset{Q\rightarrow 0}{\lim}\left( e^{-xy/Q} P(x,\widehat y\, ) \cdot e^{xy/Q}\right)
\end{eqnarray}

and therefore the $\mathcal W({\mathfrak{sl}_d})$-symmetric conformal field theory quantizes the  Fuchsian differential system corresponding to this classical limit $\mathcal W_0({\mathfrak{sl}_d})$.
\end{theorem}

More explicitely, putting $Q=0$ in the Ward identities, the only remaining term of the expression of last proposition is equal to

\begin{eqnarray}
\big\langle\big\langle \bold W_0^{(k)}(x)\bold J(\widetilde x_1)\cdots\bold J(\widetilde x_n)\big\rangle\big\rangle = \sum_{1\leq i_1<\dots< i_k\leq r} \big\langle\big\langle:\left(\bold J_{i_1}\cdots\bold J_{i_k}\right) (\widetilde x):\bold J(\widetilde x_1)\cdots\bold J(\widetilde x_n)\big\rangle\big\rangle \nonumber\\
\end{eqnarray}

that is the sum over all possible ways to insert $k$ of the $d$ copies of the chiral current $\bold J$. This is exactly what one would obtain by writing the Ward identities for a Casimir algebra-symmetric conformal field theory and we can read the loop equations on the right hand side \cite{BER1}.

\section{Quantum geometry}

We shall now define the quantum geometry associated to the $\mathcal W({\mathfrak{sl}_d})$-symmetric conformal field theory we are considering.  It first consists in a definition of the quantum spectral curve and quantum complex structure, encoded in the 2-points function, when we assume the existence of a topological regime. We will then introduce the relevant topological recursion and show one of the main results of the chapter, namely that it solves the Ward identities.

We will for simplicity restrict ourselves to the case where the Riemann surface is the Riemann sphere $\Sigma=\mathbb C\mathbb P^1$, although most of the reasoning is local and could be generalized to an arbitrary Riemann surface. Let us just remark that in this case, the generalized Ward identities of the W-algebra would actually express that the differential operator $\widehat{\mathcal E}$ given by the quantum Miura transform defines an \textit{oper}.

\subsection{Topological regime and quantum spectral curve}

This is the main assumption of this study. Let us suppose that all the functions appearing in our construction are now formal series in an expansion parameter $\varepsilon\longrightarrow 0$. 

This corresponds to a so-called \textit{heavy limit} where we rescale all the charges simultaneously $\alpha_j\longmapsto \frac 1\varepsilon \alpha_j$. This is equivalent to rescaling the chiral current $\bold J\longmapsto\frac 1\varepsilon \bold J$ and ultimately this could be reabsorbed in a redefinition of the background charge $Q\longmapsto \widetilde Q\underset{def}{=}   Q/\varepsilon$. We however consider the limit $\varepsilon\longrightarrow 0$ keeping $Q$ fixed.

Let us now assume that the chiral correlation functions with current insertions admit $\varepsilon\longrightarrow 0$ asymptotic expansions of the form

\begin{eqnarray}
\big\langle\big\langle \bold J_{i_1}(\widetilde x_1)\cdots\bold J_{i_n}(\widetilde x_n)\big\rangle\big\rangle \underset{def}{=} \sum_{g=0}^\infty \varepsilon^{2g-2+n} \left(W_{g,n}(\overset{i_1}{x_1},\dots, \overset{i_n}{x_n}) -\delta_{n,2}\delta_{g,0} \frac{(h_{i_1},h_{i_2})}{(x_1-x_2)^2}\right) \nonumber
\end{eqnarray}
\begin{eqnarray}
~
\end{eqnarray}

for all $n\in\mathbb N^*$ for which $W_{g,n}$, with $g\in\mathbb N$, is a multi-valued meromorphic function on $n$ copies of the universal covering of the punctured Riemann sphere. Such expansions define a topological regime.

\begin{remark}
Note that in this article, we make the choice of not writing explicitly the universal covering dependence of the functions appearing as coefficients of topological expansions to lighten notations. The reader should nevertheless keep in mind that these coefficients are defined on the quantum covering whose points are locally described as pairs $(\widetilde x,i)$, denoted $\overset{i}{x}$ when appearing as arguments.
\end{remark}

As a consequence, the differential operators obtained by inserting $\bold{\widehat{\mathcal E}}(x)$ into a chiral correlation function with current insertions $\langle\langle \bold J(\widetilde x_1)\cdots\bold J(\widetilde x_n)\rangle\rangle$ also admit asymptotic expansions of a similar form

\begin{eqnarray}
\left\langle\left\langle \bold{\widehat{\mathcal E}}(x)\, \bold J_{i_1}(\widetilde x_1)\cdots\bold J_{i_n}(\widetilde x_n)\right\rangle\right\rangle &\underset{def}{=}&\sum_{g=0}^\infty \varepsilon^{2g-1+n}\bold{\mathcal E}_{n}^{(g)}(x;\overset{i_1}{x_1},\dots, \overset{i_n}{x_n})\\
&\underset{def}{=}& \sum_{k=0}^d\sum_{g=0}^\infty (-1)^{d-k}\varepsilon^{2g-1+n}\nonumber\\
&& \qquad\quad\quad\times \quad P_{n,d-k}^{(g)}(x;\overset{i_1}{x_1},\dots, \overset{i_n}{x_n})\widehat y^{\, k}\nonumber\\
\end{eqnarray}

\begin{definition}
The differential operator $\bold{\mathcal E}\underset{def}{=} \bold{\mathcal E}_0^{(0)}$ is called the quantum spectral curve.
\end{definition}

The Ward identities extended to these formal $\varepsilon$-expansions imply that  the operator $\bold{\mathcal E}_n^{(g)}(x;\overset{i_1}{x_1},\dots, \overset{i_n}{x_n})$ is, for all $g,n\in\mathbb N$, $n\neq 0$, a meromorphic function of the variable $x\in\mathbb C$ with possible singularities at $x=\infty$, $x=x_i$ for some $i\in\{1,\dots,d\}$ or $x=z_j$ for some $j\in\{1,\dots,M\}$ and nowhere else.

This definition can be interpreted as exhibiting the quantization of a classical spectral curve, perturbatively this time. Indeed , define a function $E$ of the variables $x,y\in\mathbb C$ by the generic assignment

\begin{eqnarray}
E(x,y)\underset{def}{=} \Symb\left(\bold{\mathcal E}_0^{(0)}(x)\right)
\end{eqnarray}

The Riemann surface defined by the equation $E(x,y)=0$, the character variety of the quantum spectral curve, embeds in $\mathbb C^2$ and defines a $d:1$ cover of the complex plane by a meromorphic projection $x:\mathcal S \longrightarrow \mathbb C$ called the classical spectral curve of the $\mathcal W({\mathfrak{sl}_d})$-symmetric conformal field theory.

Recall that to any classical integrable systems presented in Lax form can be associated its corresponding spectral curve, a meromorphic covering of complex curves, and that this is the starting point to run the topological recursion of \cite{BEO},\cite{EO} in order for example to compute recursively the expansion coefficients of generating functions of derivatives of the $\tau$-function. We wish to upgrade these techniques to the non-commutative, or quantum, case using this operator formalism arising from conformal field theory.

\subsection{Fermionic description and notion of sheets}

A dual point of view to describe the quantum spectral curve is through the solutions of the linear differential equation it defines. Let us therefore consider a set of $d$ independent (multi-valued) functions $\psi_j$, for $j\in\{1,\dots,d\}$, satisfying

\begin{eqnarray}
\psi_j\cdot\bold{\mathcal E}=0
\end{eqnarray}

where the differential operator acts from the right.

\begin{theorem}
The quantum spectral curve decomposes as

\begin{eqnarray}
\bold{\mathcal E}=(\widehat y-Y_1( \widetilde x))\cdots(\widehat y-Y_d(\widetilde x))
\end{eqnarray}

where for any $\mu\in\{1,\dots,d\}$, 

\begin{eqnarray}
Y_\mu\underset{def}{=} Q\partial \left( \ln \frac{D_{\mu-1}}{D_\mu}\right)=\frac{Q\partial D_{\mu-1}}{D_{\mu-1}}-\frac{Q\partial D_\mu}{D_\mu}
\end{eqnarray}

with 
\begin{eqnarray}
D_\mu\underset{def}{=}\underset{0\leq i,j\leq \mu-1}{ \Det}\left((-Q\partial)^i\psi_{j+1}\right)
\end{eqnarray}

and the convention that $D_0\underset{def}{=}1$.
\end{theorem}

\begin{remark}
Let us mention that such a factorization is in general not unique and that each of these multi-valued factors have monodromies canceling in such a way that the resulting product is a well-defined differential operator on $\Sigma-\{z_1,\dots,z_M\}$. In particular, it does not depend on a choice of pre-image $\widetilde x$ of $x\in\Sigma-\{z_1,\dots,z_M\}$ in the universal covering.
\end{remark}

\begin{proof}{
If we define recursively the $Y_\mu$'s such that for any $\nu\leq \mu$,

\begin{eqnarray}
\psi_\nu \cdot (\widehat y-Y_1( \widetilde x))\cdots (\widehat y -Y_\mu(\widetilde x))=0
\end{eqnarray}

then the wanted result is a straightforward corollary of the following

\begin{lemma}
For any $\mu\in\{1,\dots,d\}$,

\begin{eqnarray}
(\widehat y- Y_1( \widetilde x))\cdots(\widehat y-Y_\mu( \widetilde x))= \underset{\mu+1}{ \Det}\begin{pmatrix}
1 & \psi_1(\widetilde x) &  \dots &  \psi_\mu(\widetilde x)  \cr
 \widehat y &  (-Q\partial)\psi_1(\widetilde x)  & \dots &  (-Q\partial)\psi_\mu(\widetilde x) \cr
\widehat y^{\, 2} &  (-Q\partial)^2\psi_1(\widetilde x) &  \dots &  (-Q\partial)^2\psi_\mu(\widetilde x)  \cr
\vdots &  &  & \vdots  \cr
\widehat y^{\, \mu} &  (-Q\partial)^\mu\psi_1(\widetilde x)  & \dots &  (-Q\partial)^\mu\psi_\mu(\widetilde x) \cr
\end{pmatrix}
 \frac {1}{D_\mu(\widetilde x)} \nonumber\\
 &
\end{eqnarray}
\end{lemma}

Indeed, since then the differential operator of the right hand side and the quantum spectral curve have same degree, leading coefficient and space of solutions, we conclude that hey are equal.
}\end{proof}

\begin{definition}{Master loop equation}\\
Define the multi-valued auxiliary operator

\begin{eqnarray}
\bold{\mathcal U}\underset{def}{=}(\widehat y-Y_2( \widetilde x))\cdots(\widehat y-Y_d(\widetilde x))
\end{eqnarray}

such that the following identity holds

\begin{eqnarray}
(\widehat y - Y_1(\widetilde x))\, \bold{\mathcal U}=\bold{\mathcal E}
\end{eqnarray}

and is called the master loop equation.
\end{definition}

\begin{remark}
This terminology comes from the connection existing between the quantum geometry of the $\mathcal W({\mathfrak{sl}_d})$-symmetric conformal field theory we are studying and that of the $\beta(\mathfrak{sl}_d)$-deformed two-matrix model, a generalization of that of \cite{BEMPF},\cite{CEO}. These similarities hide a full correspondence between the theories whose description should appear soon in the sequel \cite{BE2} of the present paper.
\end{remark}

\begin{remark}
The geometrical interpretation of this factorization is that there exists a function $Y$ that is actually multi-valued on the punctured sphere in such a way that $Y(\overset{i}{x})=Y_i(\widetilde x)$ is its value at a generic point $\widetilde x\in\widetilde\Sigma$ taken in the sheet $i\in\{1,\dots,d\}$. The sheet labeled by $i=1$ is often called the \textit{physical sheet}. The quantum sheets therefore label solutions of the quantum spectral curve and in this sense, the function $Y$ is uniquely valued on the quantum spectral curve.
\end{remark}

\subsection{2-points function}

Recall that the 2-points function has the asymptotic expansion

\begin{eqnarray}
\big\langle\big\langle \bold J_{i_1}(\widetilde x_1)\bold J_{i_2}(\widetilde x_2)\big\rangle\big\rangle = \sum_{g=0}^\infty \varepsilon^{2g} \left(W_{g,2}(\overset{i_1}{x_1},\overset{i_2}{x_n}) - \delta_{g,0}\frac{(h_{i_1},h_{i_2})}{(x_1-x_2)^2}\right)
\end{eqnarray}

In the classical formalism of \cite{BEO},\cite{EO}, the initial data needed to run the topological recursion included a symmetric bi-differential $\omega_{0,2}$ on two copies of this curve having a double pole with no residue,  bi-residue equal to 1 on the diagonal divisor and no other singularities (hence the level 1 assumption for the affine Lie algebra). A natural candidate then was the Bergman kernel, or second-kind fundamental form, associated to a choice of Torelli marking, that is to a symplectic basis of real codimension 1 homology cycles. The wanted singularities plus the requirement of vanishing periods on a given half of the symplectic basis fixes the Bergman kernel uniquely. The

The quantum setup under here under study is a direct generalization of that of both \cite{CEM} and \cite{CER} where the constructions can be interpreted as solving the Ward identities of a $\mathcal W({\mathfrak{sl}_2})$-symmetric conformal field theory. In \cite{CEM}, the existence of a hyperelliptic involution simplified the discussion and a structure of quantum Riemann surface with cuts, cycles, holomorphic differential forms and their mutual pairing was defined. In particular, a quantum Bergman kernel was introduced and its periods were vanishing on a half of a symplectic basis of quantum cycles.

In the classical limit $Q\longrightarrow 0$, the problem reduces to the computation of a conformal block of a Casimir algebra \cite{BER1}. The algebraic curve embedded in $T^*\mathbb C=\mathbb C^2$ that one obtains by taking the symbol of the quantum spectral curve of the previous section then contains the initial data needed to run, when it applies, the topological recursion. That is to say that even though $W_{0,1}$ and $W_{0,2}$ seem to be more complicated conformal blocks than the one we wish to compute from the conformal field theory point of view, they can actually be extracted directly from the data of the curve. In particular, the singular part of $W_{0,2}$ can be expressed rationally in terms of the coefficients of the equation defining the curve \cite{E2} and its regular part as a bilinear of a basis of holomorphic differentials.

This algebro-geometric construction is expected to extend to the (quantum) setup here and will be further detailed in subsequent work. By anticipation, we assume that the data of this two-point function is contained in the quantum spectral curve.

\pagebreak

\begin{definition}{Quantum Bergman kernel}\\
We interpret the leading order of the 2-points function of the theory as the quantum Bergman kernel, or second-kind fundamental form on the quantum spectral curve, and denote it by

\begin{eqnarray}
B\underset{def}{=} W_{0,2}
\end{eqnarray}

Accordingly, we define the third-kind differential form $G$ (up to a function of the variable $\overset{i}{x}$ that will play no role in what follows) by the formula

\begin{eqnarray}
\partial_z G(\overset ix,\overset jz) =  B(\overset ix,\overset jz)
\end{eqnarray}
\end{definition}

\section{Topological recursion}

\subsection{Ward identities in the $\varepsilon\longrightarrow 0$ expansion}

To solve the Ward identities recursively, we must rewrite them order by order in the topological regime. An exceptional feature of the structure of these equations is that generically only the two lowest order Ward identities, that we will call \textit{linear and quadratic loop equation}, are needed to reconstruct the chiral correlation functions with current insertions perturbatively. This illustrates the \textit{over-determination} of integrable systems. Classically, the non-generic cases correspond to those where the spectral curves exhibits non-simple ramification points and one then needs the more general formalism of \cite{BouE}.

Recall that using the multi-sheet notation, we denoted the insertion of the $k^{th}$ $\mathcal W({\mathfrak{sl}_d})$-algebra generator $\bold W^{(k)}$ at a generic point $x\in\Sigma$ into a chiral correlation function with $n\in\mathbb N^*$ current insertions at generic points $x_1,\dots,x_n\in\Sigma$ by

\begin{eqnarray}
\big\langle\big\langle \bold W^{(k)}(x)\bold J_{i_1}(\widetilde x_1),\dots,\bold J_{i_n}(\widetilde x_n)\big\rangle\big\rangle \underset{def}{=} \sum_{g=0}^\infty (-1)^k\varepsilon^{2g-1+n} P^{(g)}_{n,k}(x;\overset{i_1}{x_1},\dots,\overset{i_n}{x_n})\nonumber\\
\end{eqnarray}

Replacing this expression in the two first conformal Ward identities yields

\begin{theorem}{Linear and quadratic loop equations}\\
The axioms of the $\mathcal W({\mathfrak{sl}_d})$-symmetric conformal field theory require the linear and quadratic loop equations in the topological limit. Namely for any choice of integers $n,g\in\mathbb N$, $n\neq 0$, and any generic choice of points and sheet indices $X=\{\overset{i_1}{x_1},\dots,\overset{i_n}{x_n}\}$, 

\begin{eqnarray}
P^{(g)}_{n;0}(x,J)=\sum_{i=1}^d W_{n+1}^{(g)}(\overset ix,X)=0
\end{eqnarray}

and
\begin{eqnarray}
P^{(g)}_{n;1}(x,X)=\sum_{1\leq i<j\leq d}\left(W_{n+2}^{(g-1)}(\overset ix, \overset jx, X) + \sum_{\underset{h+h'=g}{I\sqcup I'=X}} W^{(h)}_{1+\# I}(\overset ix,I)W^{(h')}_{1+\# I'}(\overset jx,I')\right) \nonumber \\
-\, Q\sum_{i=2}^d (i-1)\partial_x W_{n+1}^{(g)}(\overset jx,X)
\end{eqnarray}

is a holomorphic functions of $x\in\Sigma-\{z_1,\dots,z_M,x_1,\dots,x_n\}$.
\end{theorem}

\begin{proof}{
The proof is a done by induction on $2g-2+n$ and it is a straightforward computation.
}\end{proof}

\subsection{Bethe roots and kernel}

Generically, a zero of $D_\mu$ for a given $\mu\in\{1,\dots,d\}$ is both a pole of $Y_\mu$ and $Y_{\mu+1}$, with residue $\pm 1$ and is not a zero of any other $D_\nu$ and therefore not a pole of any other $Y_\nu$. This statement is the quantum analog to that saying that generically, there are only two sheets meeting at a branch point of an algebraic curve.

\begin{definition}{Bethe roots}\\
Let us denote by $S_\mu \underset{def}{=}\{s\in\mathbb C\, |\, D_\mu(s)=0\}$ the set of all roots of $D_\mu$ for a given $\mu\in\{1,\dots,d\}$ and by $S \underset{def}{=}\bigcup_{\mu=1}^d S_\mu$, and call them the \textit{Bethe roots}. We will moreover generically denote $\mu_s$, for any root $s\in S$, the sheet index such that there exists exactly two functions $Y_{\mu_s}$ and $Y_{\mu_s+1}$, of which $s$ is a pole.
\end{definition}

\begin{definition}{Recursion kernel}\\
Let us define the \textit{recursion kernel} $K_\mu(\overset{i_0}{x_0},x)$ as the solution of the following differential equation which is analytic at the zero $x=s$ of $D_\mu$ ($\mu_s=\mu$):

\begin{eqnarray}
\left(Y_{\mu+1}(\widetilde x)- Y_{\mu}(\widetilde x)+
Q \partial_x \right) \, K_{\mu}(\overset{i_0}{x_0},\widetilde x)
= \frac 12 \left( G(\overset{i_0}{x_0},\overset{\mu+1}{x})-G(\overset{i_0}{x_0},\overset{\mu}{x})\right)
\end{eqnarray}

Equivalently, replacing the expressions for the $Y_\mu$'s yield

\begin{eqnarray}
\left(- 2\frac{Q\partial D_{\mu}}{D_{\mu}}+\frac{Q\partial D_{\mu+1}}{D_{\mu+1}}+\frac{Q\partial D_{\mu-1}}{D_{\mu-1}}+ Q\partial_x\right) \,K_{\mu}(\overset{i_0}{x_0},\widetilde x)
= \frac 12\left(G(\overset{i_0}{x_0},\overset{\mu+1}{x})-G(\overset{i_0}{x_0},\overset{\mu}{x})\right)\nonumber\\
\end{eqnarray}

\end{definition}

\begin{remark}
Such a solution that is analytic at $s$ does not necessarily exists and this is a consequence of the existence of a loop insertion operator in the conformal field theory and of the differential Hirota identities satisfied by the $D_\mu$'s. We postpone the proof of this claim to the appendix but drop the upper-script in $\widetilde x$ to simply write $x$.
\end{remark}

\begin{remark}
The function $K_\mu$ thus defined is not unique as one can add $f(\overset{i_0}{x_0}) \frac{D_\mu(x)^2}{D_{\mu-1}(x)D_{\mu+1}(x)}$ for any function $f$. As we shall see, our main theorem is independent of such a choice.
\end{remark}

\subsection{Recursion}

\begin{theorem}{Reconstruction by topological recursion}\\
We have the topological recursion
\begin{eqnarray}
W_{n+1}^{(g)}(\overset{i_0}{x},X)
&=& \sum_\mu \frac{1}{2\pi i} \oint_{S_\mu}  K_{\mu}(\overset{i_0}{x_0},x)\, \Big(
 W_{n+2}^{(g-1)}(\overset{\mu+1}{x},\overset{\mu}{x},X)  \nonumber \\
 &&\qquad  + \sum'_{\underset{h+h'=g}{I\sqcup I'=X}}
  W_{1+\# I}^{(h)}(\overset{\mu+1}{x},I)\,W_{1+\# I'}^{(h')}(\overset{\mu}{x},I') \Big) \end{eqnarray}
where $\oint_{S_\mu}$ means integrating along a contour that surrounds all Bethe roots in $S_\mu$ but not any of the other Bethe roots $S-S_\mu$. 
When there is a finite number of roots or when the sum $\sum_{s\in S}$ can be  defined, we may write
\begin{eqnarray}
W_{n+1}^{(g)}(\overset{i_0}{x},X)
&=& \sum_{s\in S} \Res_{x\to s} K_{\mu_s}(\overset{i_0}{x_0},x)\, \Big(
 W_{n+2}^{(g-1)}(\overset{\mu_s+1}{x},\overset{\mu_s}{x},X)  \nonumber \\
 &&\qquad  + \sum'_{\underset{h+h'=g}{I\sqcup I'=X}}
  W_{1+\# I}^{(h)}(\overset{\mu_s+1}{x},I)\,W_{1+\# I'}^{(h')}(\overset{\mu_s}{x},I') \Big)
\end{eqnarray}
thus taking the form of the topological recursion for classical spectral curves defined as finite coverings and whose branch-points are analogous to the Bethe roots.
\end{theorem}

\begin{proof}{
Let us compute the expression

\begin{eqnarray}
 \sum_{s\in S} \Res_{x\to s} K_{\mu_s}(\overset{i_0}{x_0},x)\, \Big(
 W_{n+2}^{(g-1)}(\overset{\mu_s+1}{x},\overset{\mu_s}{x},X) 
  + \sum'_{\underset{h+h'=g}{I\sqcup I'=X}}
  W_{1+\# I}^{(h)}(\overset{\mu_s+1}{x},I)\,W_{1+\# I'}^{(h')}(\overset{\mu_s}{x},I') \Big) 
\end{eqnarray}

where $\sum'_{h,h',I,I'}$ means that we exclude the cases $(h=0,I=\emptyset)$ and $(h'=0,I'=\emptyset)$ from the sum.
Let us define the same quantity as the one between parentheses but without the prime symbol $'$ :

\begin{eqnarray}
\mathcal Q_{\mu,\nu} \underset{def}{=} W_{n+2}^{(g-1)}(\overset{\mu}{x},\overset{\nu}{x},X) 
  + \sum_{\underset{h+h'=g}{I\sqcup I'=X}}
  W_{1+\# I}^{(h)}(\overset{\mu}{x},I)\,W_{1+\# I'}^{(h')}(\overset{\nu}{x},I')
\end{eqnarray}

for sheet indices $\mu,\nu\in \{1,\dots,d\}$. We thus have to compute:

\begin{eqnarray}
  \Res_{x\to s} K_{\mu_s}(\overset{i_0}{x_0},x)\, \Big(
 \mathcal Q_{\mu_s,\mu_s+1}\, -W^{(0)}_1(\overset{\mu_s+1}{x}) W^{(g)}_{n+1}(\overset{\mu_s}{x},X) - W^{(0)}_1(\overset{\mu_s}{x}) W^{(g)}_{n+1}(\overset{\mu_s+1}{x},X) \Big) 
\end{eqnarray}

Let us rewrite

\begin{eqnarray}
2 \mathcal Q_{\mu_s,\mu_s+1}
&=& \mathcal Q_{\mu_s+1,\mu_s}+\mathcal Q_{\mu_s,\mu_s+1} \cr
&=& \sum_{i\neq j} \mathcal Q_{i,j} - \sum_{j\neq \mu_s,\mu_s+1} (\mathcal Q_{\mu_s,j}+\mathcal Q_{\mu_s+1,j}) - \sum_{i\neq \mu_s,\mu_s+1} (\mathcal Q_{i,\mu_s}+ \mathcal Q_{i,\mu_s+1})  \cr
&& -  \sum_{i\neq j, \, i\neq \mu_s,\mu_s+1,\, j\neq \mu_s,\mu_s+1} \mathcal Q_{i,j} \cr
&=& 2 \sum_{i<j} \mathcal Q_{i,j} 
-  \sum_{i\neq j, \,(i,j)\neq(\mu_s,\mu_s+1),\, (i,j)\neq(\mu_s+1,\mu_s)} \mathcal Q_{i,j} \cr
&=& 2 P^{(g)}_{n;1}(x,X) + 2Q\sum_j j\,\partial_x\,W_{n+1}^{(g)}(\overset{j}{x},X) -  \sum_{\underset{(i,j)\neq(\mu_s,\mu_s+1),\, (i,j)\neq(\mu_s+1,\mu_s)}{i\neq j}} \mathcal Q_{i,j} \cr
&=& +\ 2Q \mu_s \,\partial_x\,W_{n+1}^{(g)}(\overset{\mu_s}{x},X)+2Q (\mu_s+1) \,\partial_x\,W_{n+1}^{(g)}(\overset{\mu_s+1}{x},X) + {\rm reg.\,at\,}s \cr
\end{eqnarray}

where we used the $i\longleftrightarrow j$ symmetry of the symbol $\mathcal Q_{i,j}$ and the linear loop equation after introducing $P^{(g)}_{n;1}$. Multiplying by the recursion kernel and taking the residue at the Bethe root $s\in S$ then implies

\begin{eqnarray}
&&   \Res_{x\to s} K_{\mu_s}(\overset{i_0}{x_0},x)\, \Big(
 \mathcal Q_{\mu_s,\mu_s+1}\, -W^{(0)}_1(\overset{\mu_s+1}{x}) W^{(g)}_{n+1}(\overset{\mu_s}{x},X) - W^{(0)}_1(\overset{\mu_s}{x}) W^{(g)}_{n+1}(\overset{\mu_s+1}{x},X) \Big) \cr
&=&  \Res_{x\to s} K_{\mu_s}(\overset{i_0}{x_0},x)\, \Big(
-W^{(0)}_1(\overset{\mu_s+1}{x}) W^{(g)}_{n+1}(\overset{\mu_s}{x},X) - W^{(0)}_1(\overset{\mu_s}{x}) W^{(g)}_{n+1}(\overset{\mu_s+1}{x},X) \cr
&& +\ Q \mu_s \,\partial_x\,W_{n+1}^{(g)}(\overset{\mu_s}{x},X)+Q (\mu_s+1) \,\partial_x\,W_{n+1}^{(g)}(\overset{\mu_s+1}{x},X)
 \Big) 
\end{eqnarray}

Since $W^{(g)}_{n+1}(\overset{\mu_s}{x},X)+W^{(g)}_{n+1}(\overset{\mu_s+1}{x},X)$ is analytic at $s$, we may rewrite:

\begin{eqnarray}
&&   \Res_{x\to s} K_{\mu_s}(\overset{i_0}{x_0},x)\, \Big(
 \mathcal Q_{\mu_s,\mu_s+1}\, -W^{(0)}_1(\overset{\mu_s+1}{x}) W^{(g)}_{n+1}(\overset{\mu_s}{x},X) - W^{(0)}_1(\overset{\mu_s}{x}) W^{(g)}_{n+1}(\overset{\mu_s+1}{x},X) \Big) \cr
&=&  \frac12  \Res_{x\to s} K_{\mu_s}(\overset{i_0}{x_0},x)\nonumber\\
&& \qquad\quad\qquad\times\ \Big(
W^{(0)}_1(\overset{\mu_s}{x}) - W^{(0)}_1(\overset{\mu_s+1}{x})
 + Q \,\partial_x \Big)  \left(W^{(g)}_{n+1}(\overset{\mu_s}{x},X) - W^{(g)}_{n+1}(\overset{\mu_s+1}{x},X) \right) \nonumber\\
\end{eqnarray}

Integrating by parts then yields 

\begin{eqnarray}
&&   \Res_{x\to s} K_{\mu_s}(\overset{i_0}{x_0},x)\, \Big(
 \mathcal Q_{\mu_s,\mu_s+1}\, -W^{(0)}_1(\overset{\mu_s+1}{x}) W^{(g)}_{n+1}(\overset{\mu_s}{x},X) - W^{(0)}_1(\overset{\mu_s}{x}) W^{(g)}_{n+1}(\overset{\mu_s+1}{x},X) \Big) \cr
&=&  -  \Res_{x\to s} \left(W^{(g)}_{n+1}(\overset{\mu_s}{x},X)-W^{(g)}_{n+1}(\overset{\mu_s+1}{x},X)\right)\, \Big(
W^{(0)}_1(\overset{\mu_s+1}{x}) - W^{(0)}_1(\overset{\mu_s}{x})
 + Q \,\partial_x \Big) \,K_{\mu_s}(\overset{i_0}{x_0},x)\nonumber\\
\end{eqnarray}

and using the defining differential equations of each $K_{\mu_s}$ we get 

\begin{eqnarray}
&&   \Res_{x\to s} K_{\mu_s}(\overset{i_0}{x_0},x)\, \Big(
 \mathcal Q_{\mu_s,\mu_s+1}\, -W^{(0)}_1(\overset{\mu_s+1}{x}) W^{(g)}_{n+1}(\overset{\mu_s}{x},X) - W^{(0)}_1(\overset{\mu_s}{x}) W^{(g)}_{n+1}(\overset{\mu_s+1}{x},X) \Big) \cr
&=&  - \frac 12  \Res_{x\to s} \left(W^{(g)}_{n+1}(\overset{\mu_s}{x},X)-W^{(g)}_{n+1}(\overset{\mu_s+1}{x},X)\right) \left(G(\overset{i_0}{x_0},\overset{\mu_s+1}{x})-G(\overset{i_0}{x_0},\overset{\mu_s}{x})\right) .
\end{eqnarray}

This implies that we have
\begin{eqnarray}
&& 2  \sum_\mu \oint_{S_\mu}  K_{\mu}(\overset{i_0}{x_0},x)\, \Big(
 \mathcal Q_{\mu,\mu+1}\, -W^{(0)}_1(\overset{\mu+1}{x}) W^{(g)}_{n+1}(\overset{\mu}{x},X) - W^{(0)}_1(\overset{\mu}{x}) W^{(g)}_{n+1}(\overset{\mu+1}{x},X) \Big) \nonumber \\
&=& \sum_\mu   \oint_{S_\mu} W^{(g)}_{n+1}(\overset{\mu}{x},X) G(\overset{i_0}{x_0},\overset{\mu}{x}) 
 + \sum_\mu   \oint_{S_\mu} W^{(g)}_{n+1}(\overset{\mu+1}{x},X) G(\overset{i_0}{x_0},\overset{\mu+1}{x}) \\
&&  - \sum_\mu  \oint_{S_\mu} W^{(g)}_{n+1}(\overset{\mu}{x},X) G(\overset{i_0}{x_0},\overset{\mu+1}{x}) 
  - \sum_\mu  \oint_{S_\mu} W^{(g)}_{n+1}(\overset{\mu+1}{x},X) G(\overset{i_0}{x_0},\overset{\mu}{x}) .
\end{eqnarray}

Now we notice that $W^{(g)}_{n+1}(\overset{\mu}{x},X)$ has poles at the points of $S_\mu$ and $S_{\mu-1}$, and $G(\overset{i_0}{x_0},\overset{\mu}{x})$ has poles at $x=x_0$ (simple pole with residue $-\delta_{i_0,\mu}$), and at the points of $S_\mu$ and $S_{\mu-1}$.
Moving the integration contours in the second line yields
\begin{eqnarray}
&& 2  \sum_\mu \oint_{S_\mu}  K_{\mu}(\overset{i_0}{x_0},x)\, \Big(
 \mathcal Q_{\mu,\mu+1}\, -W^{(0)}_1(\overset{\mu+1}{x}) W^{(g)}_{n+1}(\overset{\mu}{x},X) - W^{(0)}_1(\overset{\mu}{x}) W^{(g)}_{n+1}(\overset{\mu+1}{x},X) \Big) \nonumber\\
&=& 2\pi i \sum_\mu   \delta_{\mu,i_0} W^{(g)}_{n+1}(\overset{\mu}{x_0},X) +\delta_{\mu+1,i_0} W^{(g)}_{n+1}(\overset{\mu+1}{x_0},X) \nonumber \\
&& - \sum_\mu \oint_{S_{\mu-1}} W^{(g)}_{n+1}(\overset{\mu}{x},X) G(\overset{i_0}{x_0},\overset{\mu}{x}) 
 - \sum_\mu   \oint_{S_{\mu+1}} W^{(g)}_{n+1}(\overset{\mu+1}{x},X) G(\overset{i_0}{x_0},\overset{\mu+1}{x})\nonumber  \\
&&  - \sum_\mu  \oint_{S_\mu} W^{(g)}_{n+1}(\overset{\mu}{x},X) G(\overset{i_0}{x_0},\overset{\mu+1}{x}) 
  - \sum_\mu  \oint_{S_\mu} W^{(g)}_{n+1}(\overset{\mu+1}{x},X) G(\overset{i_0}{x_0},\overset{\mu}{x}) \\
&=& 4\pi i \  W^{(g)}_{n+1}(\overset{i}{x_0},X)  \nonumber \\
&& - \sum_\mu \oint_{S_{\mu}} \left( W^{(g)}_{n+1}(\overset{\mu}{x},X) +W^{(g)}_{n+1}(\overset{\mu+1}{x},X) \right) \ \left( G(\overset{i_0}{x_0},\overset{\mu}{x}) + G(\overset{i_0}{x_0},\overset{\mu+1}{x})  \right).
\end{eqnarray}
The last term vanishes because it has no pole at $S_\mu$ from the linear loop equation thus proving the theorem.
}\end{proof}

\begin{remark}
The parameters of the conformal block although absent of this computation, are encoded in the quantum spectral curve that serves as input for the topological recursion, itself being a universal recursive procedure. Indeed, the insertion points of vertex operators and the corresponding charges appear in the position of singularities of the differential operator and in the corresponding monodromy of its solutions.
\end{remark}

\subsection{Special geometry and free energies}

The topological recursion procedure applied to a quantum spectral curve allows for the perturbative reconstruction of solutions to the generalized Ward identities associated to the algebra $\mathcal W({\mathfrak{sl}_d})$. They are conformal blocks with current insertions and our last step will be to remove these currents to obtain chiral correlation functions of the form $\left\langle \prod_{j=1}^M V_{\alpha_j}(z_j)\right\rangle$ viewed as functions on the moduli space of all such W-symmetric conformal field theories on the Riemann sphere

\begin{eqnarray}
\mathcal M_{\mathcal W({\mathfrak{sl}_d})}^{\mathbb P^1}\underset{def}{=}\left\{\left(\{(z_j,\alpha_j)\in \mathbb C \mathbb P^1\times\mathfrak h^*\}_{1\leq j\leq M}, \bold t\right)\right\}
\end{eqnarray}

where the \textit{times} $\bold t\underset{def}{=}\{t_0,t_1,t_2,\dots\}$ generically denote the moduli of the potential $V\underset{def}{=}\sum_{k=0}^\infty t_k x^k$ such that $\partial V=W_{0,1}+Y$. Let us now describe how special geometry, taking the form of Seiberg-Witten relations \cite{CER}, allow us to carry on our computation.

\begin{definition}{$\tau$-function of the theory}\\
Define the $\tau$-function $\mathfrak T_{\bold z, \bold \alpha, \bold t }$ associated to the W-symmetric conformal field theory as the function on the moduli space $\mathcal M_{\mathcal W({\mathfrak{sl}_d})}^{\mathbb P^1}$ by 

\begin{eqnarray}
\ln\, \mathfrak T_{\bold z, \bold \alpha, \bold t }\underset{def}{=}\sum_{g=0}^\infty \varepsilon^{2g-2} \mathcal F_g 
\end{eqnarray}

where the genus $g$ free energy $\mathcal F_g$ is such that for any deformation $\delta\in T^*\mathcal M_{\mathcal W({\mathfrak{sl}_d})}^{\mathbb P^1}$, we have the special geometry relations 

\begin{eqnarray}
\delta\mathcal F_g=\int_{\delta^*} W_{g,1} \qquad \text{and} \qquad \delta W_{g,n} = \int_{\delta^*} W_{g,n+1}
\end{eqnarray}

for any $g,n\in\mathbb N$, $n\neq 0$, where we introduced a $generalized\ cycle$ $\delta^*$ dual to the deformation and defined by the requirement

\begin{eqnarray}
\delta W_{0,1} = \int_{\delta^*} W_{0,2}
\end{eqnarray}

\end{definition}

This needs a systematic definition and study of deformations and associated cycles that should generalize the one introduced in \cite{CEM} and \cite{CER}. In particular, the additional choices that have to be made to define form-cycle type dualities, such as pair of point decomposition of punctured surfaces, should correspond to parameters of bases of conformal blocks. The natural conjecture

\begin{conjecture}{$\tau$-function as conformal block}\\
\begin{eqnarray}
\mathfrak T_{\bold z, \bold \alpha, \bold t}=\left\langle \prod_{j=1}^M V_{\alpha_j}(z_j)\right\rangle
\end{eqnarray}
\end{conjecture}

would then go one step further in performing the conformal bootstrap of W-symmetric conformal field theories. Indeed, the remaining problem would then be to recollect these conformal blocks into single-valued smooth correlation functions. 

It is still a conjecture as one would need to examine all possible directions of deformation and at the moment, the dependence in the times $\bold t$ is yet to be understood (in relation with the determination of bases of Toda conformal blocks \cite{CPT}).

To motivate it, let us examine formally how the chiral conformal block $\langle \prod_{j=1}^M V_{\alpha_j}(z_j)\rangle$ deforms under variations of the moduli corresponding to the position of the inserted vertex operators and corresponding charges. Recall that the vertex operators of Toda conformal field theory are taken to be W-algebra primaries with respect to the W-algebra and are related to the defining chiral current by

\begin{eqnarray}
V_{\alpha_j}(z_j)\underset{def}{=}: \exp \left(\left(\alpha_j,\varphi_{|\mathfrak h_j}\right)\right):
\end{eqnarray}

in terms of the corresponding free bosons $\varphi_{|\mathfrak h_j}\underset{def}{=}\left(\int^{z_j}\bold J(x\cdot E)dx\right)_{E\in\mathfrak h_j}$. This implies the derivatives

\begin{eqnarray}
\frac{\partial}{\partial z_i}V_{\alpha_j}(z_j) &=&\delta_{i,j} : \left(\alpha_j,\bold J(z_j)\right) V_{\alpha_j}(z_j):\\
\text{and}\qquad\frac{\partial}{\partial \alpha_i}V_{\alpha_j}(z_j) &=&\delta_{i,j} :  \varphi_{|\mathfrak h_j} V_{\alpha_j}(z_j): 
\end{eqnarray}

immediately yielding

\begin{eqnarray}
\frac{\partial}{\partial z_i} \ln \mathfrak T_{\bold z, \bold \alpha, \bold t} &=&  \frac{\left\langle :(\alpha_i,\bold J(z_i))V_{\alpha_i}(z_i): \prod_{\underset{j\neq i}{j=1}}^M V_{\alpha_j}(z_j)\right\rangle}{\left\langle\prod_{j=1}^M V_{\alpha_j}(z_j)\right\rangle}\\
&\underset{def}{=}&  \underset{x=z_i}{\text{ev}}^{reg} \left(\alpha_i,{W_1}_{|\mathfrak h_i}\right)\\
\text{and}\qquad \frac{\partial}{\partial\alpha_i}\ln \mathfrak T_{\bold z, \bold \alpha, \bold t} &=&  \frac{\left\langle :\varphi_{|\mathfrak h_i}V_{\alpha_i}(z_i): \prod_{\underset{j\neq i}{j=1}}^M V_{\alpha_j}(z_j)\right\rangle}{\left\langle\prod_{j=1}^M V_{\alpha_j}(z_k)\right\rangle}\\
&\underset{def}{=}& \int^{z_i}_{reg} W_1(x)_{|\mathfrak h_i}dx
\end{eqnarray}

where the regularized linear evaluation operator and the integral are defined in the only natural way by taking the normal ordered product of the operators located at coinciding insertion points. The symbol ${W_1}_{|\mathfrak h_i}$ appearing in both relations is defined as the image of $W_1(x)\underset{def}{=}\sum_{a=1}^d W(\overset{a}{x})h_a\in\mathfrak h^*$ by the ismorphism $\mathfrak h^*\simeq \mathfrak h_i^*$.

Moreover, the wave function reconstructed from topological recursion applied to the quantum spectral curve of the $\mathcal W({\mathfrak{sl}_d})$-symmetric conformal field theory would then be defined as

\begin{eqnarray}
\Psi(D;\varepsilon)\underset{def}{=} \exp\left(\sum_{n=0}^\infty\sum_{g=0}^\infty \frac{\varepsilon^{2g-2+n}}{n!}\int_D\cdots\int_D W_{g,n} \diff x_1\cdots \diff x_n\right)
\end{eqnarray}

and is expected to be related to correlation functions of the theory with insertions of degenerate fields, in the sense that it should satisfy a KZ type equation.

\section{Conclusion}

We have proved that the Ward identities satisfied by chiral correlation functions of $\mathcal W({\mathfrak{sl}_d})$-symmetric conformal field theories can be solved perturbatively using a generalization of the topological recursion of \cite{BEO},\cite{EO} that applies to quantum curves. This can also be viewed as a quantization of the setup of \cite{BEM1},\cite{BEM2} and\cite{BER1} and this corresponds to having $Q\neq 0$ in the conformal field theory \cite{DV}. In turn, this yielded a conjecture on how to compute the topological (heavy-limit) expansion of the chiral $M$-point functions of the theory. The sequel to this article will provide a matrix model realization of this theory and later on a generalization of the scheme to arbitrary Riemann surfaces.

\section*{Aknowledgements}

BE was supported by the ERC Starting Grant no. 335739 “Quantum fields and knot homologies” funded by the European Research Council under the European Union’s Seventh Framework Programme. BE is also partly supported by the ANR grant Quantact: ANR-16-CE40-0017. 
We thank S. Ribault, J. Hurtubise and I. Runkl for fruitful discussions.

\section{Appendix}

\subsection{Loop insertion operator}

As is customary in conformal field theory, one can define a linear operator acting on the space of chiral correlation functions with current insertions consisting in the insertion of an additional current at a given point on the surface and in a given quantum sheet.

\begin{definition}{Loop insertion operator}\\
Denoting $\Sigma'\underset{def}{=}\Sigma-\{z_1,\dots,z_M\}$ and by $\widetilde\Sigma'$ its universal covering, for any $n\in\mathbb N$, any $\widetilde x_1,\dots,\widetilde x_n\in\widetilde\Sigma'$ and any sheet indices $i_1,\dots,i_n\in\{1,\dots,d\}$, define the action of the loop insertion operator at the point $\widetilde x\in\widetilde\Sigma'$ in sheet $i\in\{1,\dots,d\}$, denoted $\delta_{\widetilde x}^i$, by

\begin{eqnarray}
\delta_{\widetilde x}^{i}\cdot\big\langle\big\langle \bold J_{i_1}(\widetilde x_1)\cdots \bold J_{i_n}(\widetilde x_n)\big\rangle\big\rangle \underset{def}{=} \big\langle\big\langle \bold J_i(\widetilde x)\bold J_{i_1}(\widetilde x_1)\cdots \bold J_{i_n}(\widetilde x_n)\big\rangle\big\rangle
\end{eqnarray}
\end{definition}

\begin{remark}
This definition of a linear loop insertion operator has no reason in general to commute with the infinite sums appearing in the topological expansions of the chiral correlation functions with current insertions. We will however assume this fact to hold such that it satisfies
\begin{eqnarray}
\delta_{\widetilde x}^i\cdot \omega_{g,n}(\overset{i_1}{x_1},\dots,\overset{i_n}{x_n})=
\omega_{g,n+1}(\overset{i}{x},\overset{i_1}{x_1},\dots,\overset{i_n}{x_n})
\end{eqnarray}
for any generic values of the arguments. In particular, 
\begin{eqnarray}
\delta_{\widetilde x_1}^{i_1}\cdot Y(\overset{i_2}{ x_2})=B(\overset{i_1}{ x_1},\overset{i_2}{ x_2})
\end{eqnarray}
\end{remark}

\subsection{Differential Hirota identities}

The factors $Y_\mu$ appearing in the factorization of the quantum spectral curve are expressed as logarithmic derivatives $Y_\mu=Q\partial \ln \frac{D_{\mu-1}}{D_\mu}$ where $D_\mu$ is the principal minor

\begin{eqnarray}
D_\mu=\begin{pmatrix}
\psi_1 &  \cdots &  \psi_\mu  \cr
 (-Q\partial)\psi_1  & \cdots &  (-Q\partial)\psi_\mu \cr
 (-Q\partial)^2\psi_1 &  \cdots &  (-Q\partial)^2\psi_\mu  \cr
\vdots   &  & \vdots  \cr
(-Q\partial)^{\mu-1}\psi_1  & \cdots &  (-Q\partial)^{\mu-1}\psi_\mu \cr
\end{pmatrix}, \qquad \mu=1,\dots,d
\end{eqnarray}

\begin{proposition}{Differential Hirota identities}\\
For any quantum sheet index $\mu\in\{1,\dots,d\}$,

\begin{eqnarray}
\frac{D_{\mu-1}D_{\mu+1}}{D_\mu ^2}&=&\frac{\diff}{\diff x}\left(\frac{\widehat D_\mu}{D_\mu}\right)
\end{eqnarray}

where $\widehat D_\mu$ is the determinant used to compute $D_\mu$ but with $\psi_\mu$ replaced by $\psi_{\mu+1}$ in the last column.
\end{proposition}

\begin{proof}{
This is a straightforward computation. Indeed, the identity of the proposition is equivalent to 
\begin{eqnarray}
D_{\mu-1}D_{\mu+1}&=&\widehat D_\mu ' D_\mu-\widehat D_\mu D'_\mu
\end{eqnarray}
which is exactly a well known result from determinant computation in the form

\begin{eqnarray}
\begin{pmatrix}
 &  & * & * \cr
 & & \vdots & \vdots \cr
* & \cdots & * & * \cr
* & \cdots & * & * \cr
\end{pmatrix} \times \begin{pmatrix}
 *& \cdots &* \cr
\vdots & & \vdots \cr
*& \cdots &* \cr
\end{pmatrix}& =& \begin{pmatrix}
&  &  & * \cr
 & &  & \vdots \cr
 &  &  & * \cr
* & \cdots & * & * \cr
\end{pmatrix} \times \begin{pmatrix}
&  & * &  \cr
 & & \vdots &  \cr
* & \cdots & * & * \cr
 &  & * & \cr
\end{pmatrix} \qquad \qquad  \nonumber \\
&&\qquad\qquad  - \begin{pmatrix}
&  &   * & \cr
 & &   \vdots & \cr
 &  &   * &  \cr
* & \cdots & * & * \cr
\end{pmatrix} \times \begin{pmatrix}
&  & *   \cr
 & & \vdots   \cr
* & \cdots & *  \cr
 &  & *  \cr
\end{pmatrix} \nonumber \\
\end{eqnarray}

where this is to be understood as an algebraic relation between determinants of matrices all obtained from the same $\mu+1 \times \mu +1$ (whose determinant is $D_{\mu+1}$, it is the one represented by the second term of the product before the equal sign) by removing pairs of columns and rows. On the picture, the stars $*$ represent columns and/or rows that are removed before taking the determinant.}
\end{proof}

\begin{remark}
This is the differential version of the $QQ$-system of difference relations as appearing in the study of quantum integrable systems.
\end{remark}

\subsection{$K_{\mu}$ has no monodromy around $S_\mu$}

Recall that we defined the kernel $K_\mu$ associated to the set $S_\mu=\{s\in\Sigma'| D_\mu(s)=0\}$ as satisfying the differential equation

\begin{eqnarray}
\left(- 2\frac{Q\partial D_{\mu}}{D_{\mu}}+\frac{Q\partial D_{\mu+1}}{D_{\mu+1}}+\frac{Q\partial D_{\mu-1}}{D_{\mu-1}}+ Q\partial_x\right) \,K_{\mu}(\overset{i_0}{x_0},\widetilde x)
= G(\overset{i_0}{x_0},\overset{\mu+1}{x})-G(\overset{i_0}{x_0},\overset{\mu}{x})\nonumber\\
\end{eqnarray}

As was noted before, this defines $K_\mu$ up to terms of the form $f(\overset{i_0}{x_0}) \frac{D_\mu(x)^2}{D_{\mu-1}(x)D_{\mu+1}(x)}$ that are irrelevent when computing the $W_{g,n}$'s. The generic solution for $K_\mu$ is therefore of the form

\begin{eqnarray}
K_\mu(\overset{i_0}{x_0},\widetilde x)&=&\frac{D_\mu(\widetilde x)^2}{D_{\mu-1}(\widetilde x)D_{\mu+1}(\widetilde x)}\int^{\widetilde x} \frac{D_{\mu-1}(\widetilde x')D_{\mu+1}(\widetilde x')}{D_\mu(\widetilde x')^2}\left( G(\overset{i_0}{x_0},\overset{\mu+1}{x'})-G(\overset{i_0}{x_0},\overset{\mu}{x'})\right) \nonumber\\
&&\qquad + \quad f(\overset{i_0}{x_0})\, \, \frac{D_\mu(x)^2}{D_{\mu-1}(x)D_{\mu+1}(x)}
\end{eqnarray}

The choice of function $f$ can be reabsorbed into a $\overset{i_0}{x_0}$-dependent choice of starting point for the integral. The holomorphic condition we wish to impose on $K_\mu$ ensures it can be integrated globally along a path surrounding the Bethe roots lying in $S_\mu$ and none of the others (the ones in $S-S_\mu$). This is equivalent to requiring that the residue of the integrand appearing in the last equality vanishes at any Bethe root $s\in S_\mu$.

\begin{theorem}{Bethe equations}\\
For any generic point $\overset{i_0}{x_0}$ in the quantum covering and any Bethe root $s\in S_\mu$,

\begin{eqnarray}
\bold{BE}_{\mu,s}(\overset{i_0}{x_0})\underset{def}{=}\underset{x=s}{\Res}\left( \frac{D_{\mu-1}(\widetilde x)D_{\mu+1}(\widetilde x)}{D_\mu(\widetilde x)^2}\left( G(\overset{i_0}{x_0},\overset{\mu+1}{x})-G(\overset{i_0}{x_0},\overset{\mu}{x})\right)\right)&=&0
\end{eqnarray}

\end{theorem}

\begin{proof}{
Recall that the loop insertion operator is such that $\delta_{\widetilde x_1}^{i_1}\cdot Y(\overset{i_2}{ x_2})=B(\overset{i_1}{ x_1},\overset{i_2}{ x_2})$. Taking a primitive on both sides of this equality yields $\delta_{\widetilde x_1}^{i_1}\cdot \ln \frac{D_{\mu-1}(\widetilde x_2)}{D_{\mu}(\widetilde x_2)}=G(\overset{i_1}{x_1},\overset{i_2}{x_2})$ and in turn

\begin{eqnarray}
G(\overset{i_0}{x_0},\overset{\mu+1}{x})-G(\overset{i_0}{x_0},\overset{\mu}{x}) = \delta_{\widetilde x_0}^{i_0} \ln \frac{D_{\mu-1}(\widetilde x)D_{\mu+1}(\widetilde x)}{D_{\mu}(\widetilde x)^2}
\end{eqnarray}

Replacing this into the expression appearing in the statement of the theorem implies

\begin{eqnarray}
\bold{BE}_{\mu,s}(\overset{i_0}{x_0})&=&\underset{x=s}{\Res} \left(  \frac{D_{\mu-1}(\widetilde x)D_{\mu+1}(\widetilde x)}{D_{\mu}(\widetilde x)^2} \delta_{\widetilde x_0}^{i_0} \ln \frac{D_{\mu-1}(\widetilde x)D_{\mu+1}(\widetilde x)}{D_{\mu}(\widetilde x)^2} \right) \\
&=&\delta_{\widetilde x_0}^{i_0}\left(\underset{x=s}{\Res}  \frac{D_{\mu-1}(\widetilde x)D_{\mu+1}(\widetilde x)}{D_{\mu}(\widetilde x)^2} \right)
\end{eqnarray}
which is indeed equal to zero by the differential Hirota identities since it allows to realize the term whose residue has to be taken as a total derivative.}
\end{proof}

\end{document}